\documentclass[11pt]{amsart}
\usepackage {amsmath, amssymb, mathscinet}
\usepackage{verbatim, color, xfrac}
\newtheorem{thm}{Theorem}[section]
\newtheorem{cor}[thm]{Corollary}
\newtheorem{lem}[thm]{Lemma}
\newtheorem{defn}[thm]{Definition}

\newtheorem{preremark}[thm]{Remark}
\newenvironment{remark}%
  {\begin{preremark}\upshape}{\end{preremark}}
\newtheorem{preproblem}[thm]{Problem}
  {\begin{preproblem}\upshape}{\end{preproblem}}
\newtheorem{preexample}[thm]{Example}
\newenvironment{example}%
  {\begin{preexample}\upshape}{\end{preexample}}

\numberwithin{equation}{section}
\newtheorem{fact}[thm]{Fact}
\numberwithin{equation}{section}

\newcommand{\del}{\bigtriangleup}

\newcommand{\inv}{^{-1}}

\newcommand{\ten}{\otimes}

\newcommand{\second}{\prime \prime }
\newcommand{\comments}[1]{}

\DeclareMathOperator{\Res}{Res}

\DeclareMathOperator{\EQ}{EQ}
\begin{document}

\title[ Quadratic differential operators
]{Quadratic differential operators, Bicharacters  and $\bullet$ Products}
\author{Iana I. Anguelova}
\author{Maarten J. Bergvelt}

\address{Anguelova:
  Department of Mathematics,  College of Charleston,
Charle\-ston SC 29424 }
\email{anguelovai@cofc.edu}
\address{Bergvelt: Department of Mathematics\\ University of Illinois\\
  Urbana-Champaign\\ Illinois 61801} \email{bergv@uiuc.edu}


\begin{abstract}
  For a commutative cocommutative Hopf algebra we study the
  relationship between a certain linear map defined via a bicharacter,
  an exponential of a quadratic differential operator and a $\bullet$ product
  obtained via twisting by a bicharacter. This new
  relationship between $\bullet$ products and exponentials of quadratic
  differential operators was inspired by studying the exponential of a
  particular quadratic differential operator introduced in \cite{MR996026} and
  used in the theory of twisted modules of lattice vertex algebras.
\end{abstract}

\maketitle
\tableofcontents

\section{Introduction }
\label{sec:intro}

\subsection{Motivation From Vertex Algebras}
\label{sec:Motivation From Vertex Algebras}

We would like to start with a few words about our motivation. This
discussion uses some basic facts about vertex algebras and their
twisted modules. The reader who is not interested in vertex algebras
can skip to the next subsection \ref{sec:Overview of the Paper}. In
the main part of the paper we will not use vertex algebras at all.

In the theory of vertex algebras (see for example \cite{MR1651389},
\cite{MR2023933} for background and more details) there is a notion
of state-field correspondence: if $V$ is a vertex algebra and $a\in V$
is a state, then the theory produces a field $a(z)=Y(a,z)\colon V\to
V((z))$. A vertex algebra carries infinitely many products for the
states: if $a,b\in V$ then we define $a_{(n)}b=\Res_{z}(z^{n} a(z)b)$.

Of particular importance is the $(-1)$ product: On the one hand it
often happens that $V$ is the span of the $(-1)$ products of a number
of generating states $v^{[i]}$, i.e., span of products
$v^{[i_{1}]}_{(-1)}v^{ [i_{2}]}_{(-1)}\dots_{(-1)}v^{[i_{n}]}$. On the
other hand the $(-1)$-product of states corresponds to the normal
ordered product of the fields of these states, i.e., the field
corresponding to a state $a_{(-1)}b$ is the normal ordered product
$\colon a(z)b(z) \colon$ of the fields $a(z)$ and $b(z)$. (See
Appendix \ref{sec:NormalOrdProd} for an introduction to normal ordered
products.)

Vertex algebras have (twisted) modules $M$ (see \cite{MR996026},
\cite{MR1272580}, \cite{MR1393116}, \cite{MR1372724}, \cite{MR1978341},
\cite{MR2172171}); in this situation we still have a state-field
correspondence $Y_M$, now from states in $V$ to (twisted) fields
$a_{M}(z)$ on $M$.  In many cases one can still define the normal
ordered product of twisted fields. But, for twisted modules, it is no
longer true that in general the normal ordered product of the twisted
fields corresponds to the $(-1)$-product of states in $V$.  And vice
versa, it is not true that the field on a twisted module corresponding
to a state $a_{(-1)}b$ is the normal ordered product of the fields
$a_{M}(z)$ and $b_{M}(z)$.

To rectify this I. Frenkel, Lepowsky and Meurman introduced in
\cite{MR996026} a very clever and unexpected modification of the
construction, in the case of a lattice vertex algebra.  They
introduced a specific quadratic differential operator $\Delta_z$
acting on the states, such that the field of $a_{(-1)}b$ on a twisted
module $M$ is given by
\[
Y_{M}(a_{(-1)}b, z)=  e^{\Delta_{z}} \colon  a_M(z)b_{M}(z)\colon.
\]
(Here the meaning of the right-hand side is that one decomposes
$e^{\Delta_{z}}(a_{(-1)}b)$ in a sum of (-1) products of generating
states and takes the sum of the normal ordered product of the
corresponding fields on $M$.)

This leads to the following question: does the normal ordered product
of fields on a twisted module correspond to a new, modified, product
of states, just as the normal ordered product on the vertex algebra
itself corresponds to the $(-1)$ product? Or equivalently: is there a
modified product of states, call it $a\bullet b$, on $V$, such that
the state field correspondence $Y_{M}\colon a\mapsto a_{M}(z)$ is a
homomorphism from $V$ with the $\bullet $ product to the space of
fields on $M$ with normal ordered product as operation? The answer
turns out to be yes: the operator $e^{\Delta_{z}}$ indeed leads to a
new product $\bullet$ on $V$ satisfying this property.

In this paper we abstract this situation, and show that in fairly
general context we have a similar relation between exponentials of
quadratic differential operators and $\bullet $ products.

\subsection{Overview of the Paper}
\label{sec:Overview of the Paper}

The definition of the $\bullet$ product was  motivated by the theory
of twisted modules of vertex algebras. We abstract
and simplify the situation, and consider the following problem:\\
  Given a commutative algebra $(M,\cdot)$ and a new commutative
  product $\bullet$ on $M$ can we find a map $\EQ\colon M\to
  M$ such that
  \[
  \EQ(a\cdot b)=\EQ(a)\bullet \EQ(b), \quad a,b\in M?
\]

Of course, in general such a map $\EQ$ will not exist. Therefore we study this problem under the assumption of an extra structure on $M$: we will assume that $M$ is a commutative and
cocommutative Hopf algebra with a bicharacter $r$.

This allows us to define, using the bicharacter $r$, a general $\bullet $
product\footnote{This is a construction similar to the one first
  introduced in the context of vertex algebras by Borcherds in
  \cite{MR1865087}, but also used in a more general context in the
  theory of Hopf algebras and quantum groups.}. Next we define a map
$\EQ$, again depending on $r$, which acts as a homomorphism between
the ordinary product and the $\bullet $ product (see Section
\ref{sec:bicharacters}).

Further, we prove that we can actually find a logarithm of the map
$\EQ$ --- a quadratic operator $\Delta_z =\mathbf{Q}$ acting on the
commutative cocommutative Hopf algebra, such that we have
$\EQ=e^{\mathbf{Q}}$. (In the paper we will use the notation
$\mathbf{Q}$ instead of $\Delta_z$ in order to avoid confusion with
the notation for the coproduct in the Hopf algebra.)\footnote{The
  $\bullet $ product in this sense resembles the Moyal star product,
  see e.g. \cite{MR1301840}. The Moyal star product is of course
  noncommutative, but it too can be defined via an exponential of a
  bi-differential operator.}

In Section \ref{sec:PolyAlg} we first prove the existence of the
logarithm for the case where $M$ is a polynomial algebra.  In the next
section \ref{sec:NonTrivGrpLike} we continue with the general case of
a commutative cocommutative Hopf algebra (i.e., when grouplike
elements are present). With some restrictions on the bicharacter $r$ we
again establish the relation between the $\bullet $ product and the
maps $\EQ$ and $e^{\mathbf{Q}}$.  The result is summarized in Theorem
\ref{Thm:GroupLike_EQ_PRop}. In section \ref{sec:The
  Frenkel-Lepowsky-Meurman Example} we show that the operator
$e^{\Delta_{z}}$ of \cite{MR996026} is a special case of what we call
$e^{Q}=\EQ$.

In Appendix \ref{sec:NormalOrdProd} we make the connection between our construction for a
commutative and cocommutative Hopf algebra and the case of Heisenberg
algebra studied by I. Frenkel, Lepowsky and Meurman.  We discuss the
definition of normal ordered product for (twisted) Heisenberg fields,
and the relation with a $\bullet$ product of the space of states.
In  Appendix \ref{sec:alternativeproof} we discuss an alternative operator description of the coproduct and some of its uses.

\section{Bicharacters and Deformations of Commutative Algebras}
\label{sec:bicharacters}

For a Hopf algebra $M$ we will denote the coproduct and the counit by
$\del$ and $\eta$, the antipode by $S$. If $a$ is an element of a Hopf
algebra we will use Sweedler's notation and write $\del (a)=\sum \
a^{\prime}\ten a^{\second}$. We often will omit the summation sign.
(For more details on Hopf algebras see for example \cite{MR1321145}).

Recall that a Hopf algebra $M$ is cocommutative if for any $m\in M$ we
have $\sum \ m^{\prime}\ten m^{\second}=\sum \ m^{\second}\ten
m^{\prime}$. A primitive element $m\in M$ is such that we have $\del
(m)=m\ten 1+1\ten m$, $\eta (m)=0$, and $S(m)=-m$.  A grouplike
element $g\in M$ is such that $\del(g)=g\ten g, \ \eta (g)=1$,
$S(g)=g^{-1}$.
\begin{defn}[\textbf{Bicharacter}(\cite{MR1865087})]\label{defn:bich}
  Let $M$ be a commutative and cocommutative Hopf algebra over
  $\mathbb{C}$ and $A$ a commutative $\mathbb{C}$ algebra. An
  $A$-valued bicharacter on $M$ is a linear map $r:M\ten M\to A$, such
  that for any $a,b,c\in M$
  \begin{align*}
    r(1\ten a)&= \eta (a) = r(a\ten 1),\\
    r(ab\ten c)&= \sum r(a\ten c^{\prime})r(b\ten c^{\second}),\\
    r(a\ten bc)&= \sum r(a^{\prime}\ten b)r(a^{\second}\ten c).
  \end{align*}
\end{defn}

If $r,s$ are bicharacters as above, we can define their product
$r\circ s$ by
\begin{equation}
  \label{eq:1}
r\circ s (a\ten b)= r(a^{\prime}\ten b^{\prime})s(a^{\second}\ten b^{\second}).
\end{equation}
We refer to $\circ$ as the convolution product of bicharacters.  It is
easy to see that the product of two bicharacters is a bicharacter. By
cocommutativity of $M$ and commutativity of $A$ the convolution
product is commutative.  The formula
\begin{equation}
  \label{eq:2}
\epsilon(a\ten b)= \eta (a) \eta(b)
\end{equation}
defines the identity bicharacter $\epsilon$.  If $S$ is the antipode
of $M$ then the formula
\begin{equation}
   \label{eq:3}
   r^{-1}(a\ten b)=r(S(a)\ten b)
\end{equation}
defines the inverse bicharacter with respect to convolution. The
bicharacters on $M$ therefore form an Abelian group with respect to
convolution.

The group of bicharacters carries an involution, $r\mapsto r^{t}$
where
\[
r^{t}(a\ten b)=r(b\ten a).
\]
A bicharacter $s$ is called symmetric if it is invariant under the
involution: $s=s^{t}$.

One point of bicharacters on $M$ is that they allow us to deform the
multiplication of $M$. So let $M$ be a commutative and cocommutative
Hopf algebra over $\mathbb{C}$ as above, $A$ a $\mathbb{C}$-algebra
and let $r:M\ten M\to A$ be an $A$-valued bicharacter on $M$. We
define a new product on $M_{A}=M\ten A$ by
\begin{equation}
m_{r}(a\ten b) = \sum a^{\prime} b^{\prime}r(a^{\second}\ten b^{\second}),
\end{equation}
for any  $a, b \in M$,  where $ab$ is the initial algebra product of
$a$ and $b$ in $M$.

\begin{lem}[\textbf{Twisting by a bicharacter} \cite{MR1865087}]
\label{lem:twistingbyBich}
The product $m_{r}$ is associative and unital (with same unit $1_{M}$).
If the bicharacter is symmetric then the new
multiplication on $M_{A}$ is commutative.
\end{lem}
In general the new multiplication on $M_{A}$ is not commutative.
\begin{defn}
\label{defn:symmetrization}
\begin{bf}(Symmetrization of a bicharacter and $\bullet $ product)\end{bf}
Starting with the bicharacter $r$ (which might not be
symmetric), define a new, symmetric, bicharacter $s$ by $s=r\circ r^{t}$,
or explicitly
\[
s(a\ten b)=r(a^{\prime}\ten b^{\prime})r(b^{\prime\prime}\ten a^{\prime\prime}),
\]
and define a (commutative) multiplication on $M$ by the twisting with
this symmetrized bicharacter:
\[
a\bullet b= a^{\prime}b^{\prime}s(a^{\prime\prime}\ten b^{\prime\prime}).
\]
In general we will call a $\bullet$ product on $M$ any twisting
$m_{s}$ of the standard product by a symmetric bicharacter.
\end{defn}
We write $\bullet_s$ when we want to emphasize the dependence on the
symmetric bicharacter $s$.

Note that if $r$ happens to be symmetric, the new bicharacter $s$
obtained from $r$ is \emph{not} identical to $r$: for instance if
$g, \bar{g}$ are grouplike, then $s(g\ten
\bar{g})=r(g\ten \bar{g})^{2}$, and if $x,y$ are
primitive $s(x\ten y)=2r(x\ten y)$.

Now we can ask what the relation is between the two multiplications on
$M_{A}$: the original multiplication, and the $\bullet $
multiplication obtained by twisting with $s$. To answer this question
we begin with the following definition:
\begin{defn}
\label{defn:TheMapEQ}
\begin{bf}(The linear map $\EQ_r$)\end{bf}
Let $M$ be a commutative and cocommutative Hopf algebra over
$\mathbb{C}$, $A$ a commutative $\mathbb{C}$-algebra and let $r:M\ten M\to A$ be
an $A$-valued bicharacter on $M$. Define a linear map
\begin{equation}
\EQ_r\colon M \to M_{A},\quad m\mapsto r(m^{\prime}\ten
m^{\prime\prime})m^{\prime\prime\prime}.\label{eq:DefEQ}
\end{equation}
Here we write $\Delta^{2}(m)=\sum m^{\prime}\ten m^{\prime\prime}\ten
m^{\prime\prime\prime}$.
\end{defn}
If the bicharacter $r$ is clear from the context, we will just write $\EQ$ for $\EQ_r$.
\begin{example}
For any bicharacter $r$ if $x$ is primitive we have
\begin{equation*}
r(x\ten 1)=r(1\ten x)=0,
\end{equation*}
thus  $\EQ_r(x)=x$ for any bicharacter $r$.
If $g$ is grouplike, then
$\EQ_r(g)=gr(g\ten g)$.
\end{example}
\begin{lem}
\label{lem:EQisHomom}
  If $a,b\in M$, where $M$ is a commutative and cocommutative Hopf algebra over
$\mathbb{C}$ as above, then for any bicharacter $r$ with
symmetrization $s$ we have
\[
\EQ_{r}(ab)=\EQ_{r}(a)\bullet_{s} \EQ_{r}(a).
\]
\end{lem}

We first recall some facts that will be used in the proof
below. Coassociativity requires
\begin{align*}
  (\del)^2 (a)&=a^{\prime}\ten a^{\prime\prime}\ten a^{\prime\prime\prime}=\\
  =(Id \ten \del) \circ \del (a)&=a^{\prime}\ten
  {(a^{\prime\prime})}^{\prime} \ten
  {(a^{\prime\prime})}^{\prime\prime}= \\ =(\del \ten Id) \circ \del
  (a) &={(a^{\prime})}^{\prime} \ten {(a^{\prime})}^{\prime\prime}\ten
  a^{\prime\prime},
\end{align*}
for any element $a$ of the Hopf algebra $M$, and similarly we can
uniquely define
\begin{equation}
  \del ^{n-1}(a)=\sum \ a^{(1)}\ten a^{(2)}\ten \dots \ten a^{(n)}
\end{equation}
(extended Sweedler notation). By cocommutativity of $M$ the factors
$a^{(i)}$ of $\Delta^{(n-1)}(a)$ are invariant under all permutations
of $n$.
\begin{proof}
  \begin{align*}
    \EQ(ab)& =r((ab)^{\prime}\ten
    (ab)^{\prime\prime})(ab)^{\prime\prime\prime}
    =r(a^{\prime}b^{\prime}\ten  a^{\prime\prime}b^{\prime\prime})
    a^{\prime\prime\prime}b^{\prime\prime\prime}=\\
    & =r(a^{\prime}\ten {(a^{\prime\prime}b^{\prime\prime})}^{\prime})
    r(b^{\prime}\ten {(a^{\prime\prime}b^{\prime\prime})}^{\prime\prime})a^{\prime\prime\prime}b^{\prime\prime\prime}=\\
    & =r(a^{(1)}\ten {(a^{(2)})}^{\prime}{(b^{(2)})}^{\prime})
    r(b^{(1)}\ten {(a^{(2)})}^{\prime\prime}{(b^{(2)})}^{\prime\prime})a^{(3)}b^{(3)}=\\
    & =r(a^{(1)}\ten a^{(2)}b^{(2)})r(b^{(1)}\ten a^{(3)}b^{(3)})a^{(4)}b^{(4)}=\\
    & =r({(a^{(1)})}^{\prime}\ten
    a^{(2)})r({(a^{(1)})}^{\prime\prime}\ten
    b^{(2)})r({(b^{(1)})}^{\prime}\ten a^{(3)})\cdot\\
&\qquad\qquad \cdot
r({(b^{(1)})}^{\prime\prime}\ten b^{(3)})a^{(4)}b^{(4)}=\\
    & =r(a^{(1)}\ten a^{(3)})r(a^{(2)}\ten b^{(3)})r(b^{(1)}\ten
    a^{(4)})r(b^{(2)}\ten b^{(4)})a^{(5)}b^{(5)}
  \end{align*}
On the other hand,
  \begin{align*}
    \EQ(a)\bullet \EQ(b)& =\big( r(a^{\prime}\ten
    a^{\prime\prime})a^{\prime\prime\prime}\big)
    \bullet \big( r(b^{\prime}\ten b^{\prime\prime})b^{\prime\prime\prime}\big)=\\
    & =r(a^{\prime}\ten a^{\prime\prime})r(b^{\prime}\ten
    b^{\prime\prime})(a^{\prime\prime\prime})
    \bullet (b^{\prime\prime\prime})=\\
    & =r(a^{\prime}\ten a^{\prime\prime})r(b^{\prime}\ten
    b^{\prime\prime}){(a^{\prime\prime\prime})}^{\prime}{(b^{\prime\prime\prime})}^{\prime}
    s({(a^{\prime\prime\prime})}^{\prime\prime}\ten {(b^{\prime\prime\prime})}^{\prime\prime})=\\
    & =r(a^{(1)}\ten a^{(2)})r(b^{(1)}\ten b^{(2)})a^{(3)}b^{(3)}
    s(a^{(4)}\ten b^{(4)})=\\
    & =r(a^{(1)}\ten a^{(2)})r(b^{(1)}\ten b^{(2)})a^{(3)}b^{(3)}
    r(a^{(4)}\ten b^{(4)})r(b^{(5)}\ten a^{(5)})
 \end{align*}
The equality then follows
using cocommutativity  of the coproduct:
\begin{align*}
  \EQ(ab)& =r(a^{(1)}\ten a^{(3)})r(a^{(2)}\ten b^{(3)})r(b^{(1)}\ten a^{(4)})r(b^{(2)}\ten b^{(4)})a^{(5)}b^{(5)}=\\
  & =r(a^{(1)}\ten a^{(2)})r(a^{(3)}\ten b^{(3)})
  r(b^{(4)}\ten a^{(4)})r(b^{(1)}\ten b^{(2)})a^{(5)}b^{(5)}=\\
  & =r(a^{(1)}\ten a^{(2)})r(b^{(1)}\ten b^{(2)}) r(a^{(4)}\ten
  b^{(4)})r(b^{(5)}\ten a^{(5)})a^{(3)}b^{(3)}=\\& =\EQ(a)\bullet \EQ(b).
  \end{align*}
\end{proof}
The conclusion is that the map $\EQ$ is a homomorphism from $(M_{A}, \cdot)$ to
$(M_{A},\bullet )$.

Notice that the algebra structure on $(M, \cdot)$ (and by extension of
$(M_{A}, \cdot)$) can be considered to be in fact the twisting of
$M$ by the identity bicharacter $\epsilon$ (see \eqref{eq:2}).

Thus $(M, \cdot)$ is in fact $(M, \bullet _{\epsilon})$, and the map
$\EQ_r$ is a homomorphism from $(M_{A}, \bullet _{\epsilon})$ to
$(M_{A},\bullet _{s})$. This leads to the question of the relation
between the different multiplication structures on $M$ given by
symmetric bicharacters $s$. If $s_{1}$ and $s_{2}$ are the
symmetrization of bicharacters $r_{1}$ and $r_{2}$ then there is a
homomorphism $EQ$ that intertwines them. Indeed, $\EQ_{-}$ is a
homomorphism from the Abelian group of bicharacters to
linear maps on $M$:
\begin{lem} For all bicharacters $r_{1},r_{2}$ we have
  \[
  \EQ_{r}=\EQ_{r_{1}}\circ\EQ_{r_{2}},\quad \text{if $r=r_{1}\circ
    r_{2}$.}
  \]
\end{lem}
\begin{proof}
\begin{align*}
  &\EQ_{r_{1}}\circ \EQ_{r_{2}}(a)=\EQ_{r_{1}}(r_{2}(a^{(1)}\ten
  a^{(2)})a^{(3)})= \\
  & =r_{1}({a^{(3)}}^{\prime}\ten {a^{(3)}}^{\prime\prime}){a^{(3)}}^{\prime\prime\prime} r_{2}(a^{(1)}\ten
  a^{(2)})=r_{1}(a^{(1)}\ten
    a^{(2)})r_{2}(a^{(3)}\ten a^{(4)})a^{(5)}=\\
    & =r_{1}\circ r_{2}(a^{\prime}\ten
    a^{\prime\prime})a^{\prime\prime\prime} =r(a^{\prime}\ten
    a^{\prime\prime})a^{\prime\prime\prime}=\EQ_{r}(a),
\end{align*}
by coassociativity and cocommutativity.
\end{proof}
\begin{cor}
\label{cor:orbit}
Each $\EQ_{r}$ is invertible, with inverse $\EQ_{r\inv}$ and if
$r_{i}$ are bicharacters with symmetrization $s_{i}$, $i=1,2$ then
\[
\EQ_{r_{2}}\circ \EQ_{r_{1}\inv}\colon (M,\bullet_{s_{1}})\to (M,\bullet_{s_{2}})
\]
gives the homomorphism from the multiplication $\bullet_{s_{1}}$ to
$\bullet_{s_{2}}$.
\end{cor}
More generally one can ask about the relation between symmetric
bicharacters that are not necessarily symmetrizations.

 Indeed, we have the
following generalization of Lemma \ref{lem:EQisHomom}.
\begin{thm}
\label{thm:EQinterchangeTh}
Let $r$ be a bicharacter defined on a commutative and cocommutative
Hopf algebra $M$ and $s_{1}$ any
symmetric bicharacter. We have
\begin{equation}
\label{eq:EQinterchangeEq}
\EQ_{r}(a \bullet _{s_1} b) = \EQ_{r}(a) \bullet _{s_{2}} \EQ_{r}(b),
\end{equation}
for  $s_{2}=s\circ s_{1}$, where $s$ is the  symmetrization  of the bicharacter $r$.
\end{thm}
\begin{proof}
The proof is very similar to the proof of the Lemma \ref{lem:EQisHomom} above.
\end{proof}
Observe that the maps $\EQ_{r}$ gives an action of the Abelian group of
bicharacters on the space of $\bullet$ products.
The theorem implies that the space of $\bullet$ products (parametrized
by symmetric bicharacters) decomposes into orbits under the action of
the Abelian group of all bicharacters.

This leads to the question of the orbit structure of the space of
$\bullet$ products on $M$.  In the next section we show that if $M$ is
generated by primitive elements $s$ is always a symmetrization, and consequently  there is only one orbit, and all $\bullet$ products are related
to the trivial product $\cdot=\bullet_{\epsilon}$. In Section
\ref{sec:NonTrivGrpLike} we discuss the case where there are also nontrivial
grouplike elements in $M$. Then there the orbit structure can be more
complicated.

\section{The Polynomial Algebra and Quadratic Differential Operators}
\label{sec:PolyAlg}

Consider a cocommutative Hopf algebra $M$ with no grouplike elements
except the unit 1.  Any such Hopf algebra over a field of
characteristic 0 is the universal enveloping algebra of the Lie
algebra of its primitive elements (for proof see for example
\cite{MR0174052}). Suppose  that $M$ is also commutative: then $M$
is the universal enveloping algebra of the \textbf{abelian} Lie
algebra of its primitive elements. Thus any such Hopf algebra over
$\mathbb{C}$ is nothing else but the polynomial algebra over a basis
for the primitive elements. The result described below works for any
polynomial algebra $M$, but for the examples we are interested in
(motivated from the theory of twisted modules of vertex algebras) we
will only look at the case when the algebra $M$ is generated by
countably many primitive generators.

Thus, consider the polynomial algebra
\[
V_0=\mathbb{C}[x_{1},x_{2},\dots]
\]
in variables $x_{i},i\ge1 , \ i \in \mathbb{N}$. ($V_0$
is the universal enveloping algebra of the Abelian Lie algebra
$\oplus_{i\ge1}\mathbb{C}x_{i}$). As such, $V_0$ is a commutative and
cocommutative Hopf algebra with primitive generators $x_{i}$.

First of all we can immediately answer the question whether a symmetric
bicharacter gives a multiplication $\bullet_{s}$ related to the
standard multiplication $\cdot$ of $V_{0}$ by an isomorphism $\EQ_{r}$ for
some bicharacter $r$.

This is certainly the case if $s$ is a symmetrization, as we argued
above, but in our present case ($V_{0}$ generated by primitives) we
see that we can define, given $s$, a bicharacter $r$ by defining $r(x_{i}\ten x_{j})=\frac{1}{2}s(x_{i}\ten x_{j})$ on the generators.
The symmetrization of the bicharacter $r$ is then
\[
r\circ r^{t}(x_{i}\ten x_{j})=
r(x_{i}^{\prime}\ten x_{j}^{\prime})r(x_{j}^{\prime\prime}\ten
x_{i}^{\prime\prime})=r(x_{i}\ten x_{j})+r(x_{j}\ten
x_{i})=s(x_{i}\ten x_{j}).
\]
\begin{thm}
 \label{thm:newinterch}
  Let $V_{0}=\mathbb{C}[x_{1},x_{2},\dots]$. Then for any symmetric
  bicharacters $s_1, s_2$   exists an isomorphism $\EQ\colon (V\ten A, \bullet_{s_1})\to
  (V\ten A,\bullet_{s_2})$  intertwining
   the  $\bullet$ products. In particular, exists an isomorphism $\EQ\colon (V\ten A, \cdot)\to
  (V\ten A,\bullet_{s})$ relating the standard product to any $\bullet_s$ product.
\end{thm}
\begin{proof}
Follows from Theorem \ref{thm:EQinterchangeTh} and Corollary \ref{cor:orbit}.
\end{proof}
Now in the present case ($V_{0}$ generated by primitives) the map
$\EQ$, although defined by a bicharacter $r$, depends only on the  symmetrization $s$ of the bicharacter $r$. In order to prove this, see
Corollary \ref{cor:Finding-r-from-s}, we first proceed to give another
description of the map $\EQ$.

The coproduct on $V_0$ is conveniently described using infinite order
differential operators. Write $x_{i}^{(1)}=x_{i}\ten 1$ and
$x_{i}^{(2)}=1\ten x_{i}$. Then we have
\begin{equation}
\label{eq:PoloCoprod}
\Delta(f(x_{i}))=f(x_{i}^{(1)}+x_{i}^{(2)})=e^{\sum_{i\ge1}
  x_{i}^{(1)}\frac{\partial}{\partial x_{i}^{(2)}}}f(x_{i}^{(2)}),
\end{equation}
and similarly
\begin{equation}
  \Delta^{2}(f(x_{i}))=f(x_{i}^{(1)}+x_{i}^{(2)}+x_{i}^{(3)})=e^{\sum_{i\ge1}
    (x_{i}^{(1)}+x_{i}^{(2)})\frac{\partial}{\partial x_{i}^{(3)}}}f(x_{i}^{(3)}),\label{eq:CoprodSquare}
\end{equation}
where $f(x_{i})$ is any polynomial in the variables $x_{i}, \ i\ge1$.

Next we fix a bicharacter on $V_0$.  Since $V_0$ is generated by the
variables $x_{i}$ the bicharacter is completely determined by the
elements
\[
q_{mn}=r(x_{m}\ten x_{n})\in A, \quad m,n\ge1.
\]
Note also that we have the following simple properties of the
bicharacter evaluated at powers of the variables:
\begin{equation}
  \label{eq:Value_r_on_variables}
  r(x_{m}^{s}\ten x_{n}^{t})=s!\delta_{st} r(x_{m}\ten x_{n})^{s}.
\end{equation}
The following lemma asserts that we can find the logarithm of the map $\EQ$:
\begin{lem}
  \label{lem:eQis EQ}
  The map $\EQ$ defined in \eqref{eq:DefEQ} is the exponential of the infinite
  order quadratic differential operator $Q_p:V_0\to V_0\ten A$,
\[
\EQ(f)=e^{Q_p}(f),\quad \text{where} \quad
Q_p=\sum_{m,n\ge1}q_{mn}\frac{\partial^{2}}{\partial x_{m}\partial
  x_{n}}.
\]
\end{lem}
\begin{proof}
We use the exponential form \eqref{eq:CoprodSquare} of the coproduct:
\begin{equation*}
  \EQ(f)=r(f^{\prime}\ten f^{\prime\prime})f^{\prime\prime\prime}=
       r(e^{\sum_{m}x_{m}\partial_{x_{m}}}\ten e^{\sum_{n}x_{n}\partial_{x_{n}}})f(x_{i}),
\end{equation*}
where the exponentials are expanded as power series and the partial
derivatives act on $f$. Next we use some simple properties of the
bicharacter: we have
\[
r(e^{\sum_{m}x_{m}\partial_{x_{m}}}\ten
e^{\sum_{n}x_{n}\partial_{x_{n}}})=\prod_{m,n}
r(e^{x_{m}\partial_{x_{m}}}\ten e^{x_{n}\partial_{x_{n}}}),
\]
and, using \eqref{eq:Value_r_on_variables},
\begin{align*}
  r(e^{x_{m}\partial_{x_{m}}}\ten e^{x_{n}\partial_{x_{n}}})&=
  \sum_{s,t\ge0} r(x_{m}^{s}\ten
  x_{n}^{t})\frac{\partial^{s+t}}{s!t!\partial^sx_{m}\partial^{t}x_{n}}=\\
  &=\sum_{s,} r(x_{m}\ten
  x_{n})^{s}
  \left(
\frac{\partial^{2}}{\partial x_{m}\partial x_{n}}
  \right)^{s}/s!=e^{r(x_{m}\ten x_{n})\frac{\partial^{2}}{\partial x_{m}\partial x_{n}}}.
\end{align*}
Combining the last two results gives the proof of the lemma.
\end{proof}

We can rephrase Lemmas \ref{lem:EQisHomom} and \ref{lem:eQis EQ} by
introducing the notion of a $\bullet$ polynomial.  If $P\in V_{0}$,
then $P$ is a linear combination of monomials $x_{i_{1}}x_{i_{2}}\dots
x_{i_{s}}$. Define then $P^{\bullet}$, the $\bullet$
polynomial\footnote{The $\bullet$ polynomial of a monomial
  $x_{i_{1}}x_{i_{2}}\dots x_{i_{s}}$ can be thought of as a "normal ordered product of states"---the analog
  of the normal ordered product of fields $a_{i_{1}}(z),
  a_{i_{2}}(z),\dots,a_{i_{s}}(z)$ in the theory of vertex
  algebras. See Appendix \ref{sec:NormalOrdProd}.}
of $P$, as the same linear combination of expressions
$x_{i_{1}}\bullet x_{i_{2}} \bullet\dots\bullet x_{i_{s}}$.

We get then the following result from Lemmas \ref{lem:EQisHomom} and \ref{lem:eQis EQ}:
\begin{thm}
\label{thm:e^QisHomom}
Let $V_0=\mathbb{C}[x_{1},x_{2},\dots]$, with $r$ a bicharacter with
values in $A$ on $V_{0}$ and let for $q_{mn}=r(x_{m}\ten x_{n})$ the
quadratic differential operator $Q_p$ be
\[
Q_{p}=\sum_{m,n\ge1}q_{mn}\frac{\partial^{2}}{\partial
  x_{m}\partial x_{n}}
\]
Let $s=r\circ r^{t}$ be the symmetrization of $r$ and
$\bullet=\bullet_{s}$ the associated $\bullet$ product on $V_{0}$.
Then for all $P\in V_{0}$ we have
\[
e^{Q_{p}}(P)=EQ_{r}(P)=r(P^{\prime}\ten
P^{\prime\prime})P^{\prime\prime\prime}.
\]
Moreover we also have
\[
e^{Q_{p}}(P)=P^{\bullet}.
\]
and $e^{Q_{p}}$ is a homomorphism from $(V_0\ten A, \cdot)$ to
$(V_0\ten A,\bullet )$.
\end{thm}
\begin{example}\label{Example:succesivebulletproduct}
  For instance: the $\bullet$ monomial of a single variable is
  \begin{equation*}
    (x_m)^{\bullet}  =e^{Q_p}(x_m)= \EQ_p(x_m) =x_m,
    \end{equation*}
    and for two variables we have
    \begin{equation*}
    (x_mx_n)^{\bullet}  = x_m \bullet  x_n =m_{s}(x_m\ten x_n)=x_m x_n +s(x_m \ten x_n) = x_m x_n +(q_{m n}+q_{n m}),
  \end{equation*}
   One should note the above is indeed true both when $m \ne n$ and  when $m=n$.

   One should also be careful when mixing the two products---we can
   write the twisting (it involves both products):
  \begin{equation*}
    (x_m x_n) \bullet  x_l =m_{s}(x_m x_n\ten x_l) = x_nx_mx_m +(q_{m l}+q_{l m})x_n +(q_{n l}+q_{l n})x_m,
  \end{equation*}
  but the above is \textbf{not} the $\bullet$ monomial $P^{\bullet} =x_m \bullet  x_n\bullet
  x_l=e^Q_p(P)$, corresponding to the monomial $P=x_mx_nx_l$; $P^{\bullet} $ is
  by definition the \textbf{successive} application of twisting by the
  bicharacter:
  \begin{equation*}
    P^{\bullet} =  x_m \bullet  x_n\bullet  x_l
    =x_nx_mx_m +(q_{m n}+q_{n m})x_l +(q_{m l}+q_{l m})x_n +(q_{n l}+q_{l n})x_m.
  \end{equation*} {}
\end{example}
\begin{example}
  To get examples of bicharacters, and hence of maps $\EQ_r$ and
  quadratic differential operators one has great freedom. One first
  chooses a $\mathbb{C}$-algebra $A$, and then for each pair of primitive
  elements $x_{m},x_{n}$ an element $q_{mn}$ in $A$. These can be
  conveniently encoded in a generating series
\[
f(x,y)=\sum_{m,n=1}^{\infty}q_{mn}x^{m}y^{n}
\]
so that
\[
r(x_{m}\ten x_{n})=q_{mn}=\frac{\partial_{x}^{m}}{m!}\frac{\partial_{x}^{n}}{n!}f(x,y).
\]
See Example \ref{example:TheBigEx} for an explicit choice of such
generating series.
\end{example}
\begin{cor}\begin{bf} (to Theorem \ref{thm:newinterch})\end{bf}
\label{cor:Finding-r-from-s}
Let $s_1, s_2: V_0 \to A$ be two symmetric bicharacters defined on the
polynomial algebra $V_0$. Then the  map $\EQ_r\colon V_0 \to
V_0\ten A$ intertwining the $\bullet$ products
\begin{equation}
\EQ_{r}(a \bullet_{s_{1}} b) = \EQ_{r}(a) \bullet_{s_{2}} \EQ_{r}(b),
\end{equation}
is  unique  among the maps $EQ_r$, for various bicharacters $r$, and depends only on the symmetrization bicharacter $s=r\circ r^t =s_2\circ s_1^{-1}$.
\end{cor}
\begin{proof}
From  Lemma \ref{cor:orbit}, we can determine the
  symmetrization bicharacter $s=s_1\circ (s_2)^{-1}$ for the map
  $\EQ_{r}$. It is obvious we cannot find the bicharacter $r$ from it's
  symmetrization $s$ uniquely, but on $V_0$ the map $\EQ_{r}$ doesn't actually depend on $r$, but
  only on its symmetrization bicharacter $s=r\circ r^t$: Since the map $\EQ_{r}$
  on $V_0$ actually coincides with the map $e^{Q}$, and $Q$ is a
  quadratic differential operator on a polynomial algebra, we can
  rewrite $Q$ as
\begin{align*}
Q &=\sum_{m,n\ge 1}q_{mn}\frac{\partial^{2}}{\partial x_{m}\partial
  x_{n}}=\sum_{m,n\ge 1}\frac{q_{mn} +q_{nm}}{2}\frac{\partial^{2}}{\partial x_{m}\partial
  x_{n}}=\\
   &=\sum_{m,n\ge 1}\frac{r(x_m \ten x_n) + r(x_n \ten x_m)}{2}\frac{\partial^{2}}{\partial x_{m}\partial
  x_{n}}=\sum_{m,n\ge 1}\frac{s(x_m \ten x_n)}{2}\frac{\partial^{2}}{\partial x_{m}\partial
  x_{n}}
  \end{align*}
Thus, we see that the map  $e^{Q}$, and so  to the map $\EQ_{r}$ on $V_0$, depends only on the symmetrization bicharacter $s$. Thus, as a map, the intertwiner $\EQ$ of the bullet products $\bullet_{s_{1}}$ and $\bullet_{s_{2}}$  is unique.
\end{proof}

\section{The  case of a Commutative and Cocommutative Hopf Algebra.}
\label{sec:NonTrivGrpLike}

In the previous section we saw that if $M$ is generated by primitive
elements any two $\bullet$ products can be intertwined (see Theorem \ref{thm:newinterch}). This depended on the fact that for a Hopf algebra $M$ generated by primitive
elements every symmetric bicharacter is a symmetrization. This is no
longer true if $M$ contains nontrivial grouplike elements (i.e., grouplike elements distinct from the unit element $1_{M}$ of $M$).

\begin{example}
  Let $M$ be the group algebra of the free rank 1 abelian group. $M$
  has a basis $e^{n\alpha}$, $n\in \mathbb{Z}$, with
  multiplication $e^{m\alpha}e^{n\alpha }=e^{(m+n)\alpha}$, and $
  e^{0}=1_{M}$. To define a bicharacter choose $A=\mathbb{C}[z]$, and
  define on the generator
\[
s(e^{\alpha}\ten e^{\alpha})=z.
\]
Clearly there is no bicharacter $r$ (with values in $A$) with a
symmetrization $s$: if such a bicharacter would exist we would have
\[
r\circ r^{t}(e^{\alpha}\ten e^{\alpha})=r(e^{{\alpha}^{\prime}}\ten
e^{{\alpha}^{\prime}})r(e^{{\alpha}^{\prime\prime}}\ten
e^{{\alpha}^{\prime\prime}})=r(e^{\alpha}\ten e^{\alpha})^{2}=z
\]
This has no solution in $A$.

Similarly, there is no isomorphism $\EQ_{r}\colon
(M_{A},\cdot)\to (M_{A},\bullet_{s})$: For we have
\begin{align*}
& \EQ_r(e^{2\alpha}) =e^{2\alpha} r(e^{2\alpha}\ten e^{2\alpha})=e^{2\alpha} r(e^{\alpha}\ten e^{\alpha})^4\\
& \EQ_r(e^{\alpha})\bullet_s \EQ_r(e^{\alpha}) =e^{2\alpha}r(e^{\alpha}\ten e^{\alpha})^2s(e^{\alpha}\ten e^{\alpha})= ze^{2\alpha}r(e^{\alpha}\ten e^{\alpha})^2.
\end{align*}
Thus if such an isomorphism $\EQ_{r}$ would exist we would have
\begin{equation*}
e^{2\alpha} r(e^{\alpha}\ten e^{\alpha})^4 =\EQ_r(e^{2\alpha})=\EQ_r(e^{\alpha})\bullet_s \EQ_r(e^{\alpha})=ze^{2\alpha}r(e^{\alpha}\ten e^{\alpha})^2 ,
\end{equation*}
which is not possible for a bicharacter $r$ taking values in $A=\mathbb{C}[z]$.
\end{example}

From now on in this paper we will assume that for all
bicharacters $r$ (symmetric or not) the values on grouplikes is a
constant, i.e.,
\[
r(g\ten \tilde g)\in \mathbb{C}\subset A,
\]
for any  $g, \tilde g$ are grouplike. In this case it is immediate  that all
symmetric bicharacters are symmetrizations, so that all $\bullet$
product are related to $\cdot=\bullet_{\epsilon}$ by some map
$\EQ_{r}$. In fact, for this to be true, we only need to require that the values of the bicharacters on grouplikes are exact squares; but we need the "constant on grouplikes" condition if we are to find logarithm of the map $EQ_r$ (see Lemma \ref{lem:eQis EQ2} and  Remark \ref{remark:reasonForLog}). Even if this "constant on grouplikes" condition is imposed,  in contrast to Corollary  \ref{cor:Finding-r-from-s}, it is clear that  the intertwiner maps $EQ_r$ are not unique when grouplike elements are present.

We require our Hopf algebras to be commutative and cocommutative.
This implies (see \cite{MR0214646} for details) that $M$ is the
product of a group algebra $\mathbb{C}[G]$, and an universal
enveloping algebra $\mathcal{U}(L)$, where $G$ and $L$ are Abelian.
For simplicity assume that $G$ is finitely generated. Then $G$
decomposes as $G=G_{\mathrm{Tor}}\times G_{\mathrm{Free}}$, with
$G_{\mathrm{Tor}}$ the torsion part, a finite group consisting of
elements $\delta$ of finite order, and $G_{\mathrm{Free}}$ the free
part, generated by elements
$\alpha_{1},\alpha_{2},\dots,\alpha_{\ell}$ of infinite order. Let $L$
have basis $x_{n}, n\ge1$. Then our Hopf algebra has the form
\[
M=\mathbb{C}[\delta_{1},\delta_{2},\dots,\delta_{s},e^{\pm\alpha_{1}},\dots,e^{\pm\alpha_{l}},x_{1},x_{2},\dots],
\]
where $\delta_{t}^{N_{t}}=1$, $e^{\alpha_{i}}e^{-\alpha_{i}}=1$. The
Hopf structure is determined by declaring the elements
$\delta_{1},\dots,\delta_{s}, e ^{\alpha_1}, \dots , e^{\alpha_\ell}$
to be grouplike, and the elements $x_{1},x_{2},\dots $ to be
primitive.

We would like to give a logarithm  of the map
$\EQ_{r}$ (similar to  Lemma~\ref{lem:eQis EQ}) in the more general case when grouplike elements are present in $M$.

But first consider the values of a bicharacter on torsion elements. If
$\delta\in M$ is finite order grouplike element, $\delta^{N}=1$, then we have for any grouplike $g$ $r(\delta^{N}\ten g)=r(\delta\ten g)^{N}=1$. So in this case $r(\delta\ten g)$ is a root of unity. In
the same way,  if $x$ is primitive we have $r(\delta\ten
x)=0$. This implies that the torsion elements contribute to $r$ and
hence to $\EQ$ only root of unity factors. Therefore we will assume
 for simplicity's sake that $M$ is in fact torsion free. (In
Remark \ref{rem:torsion_elements} we will discuss the impact of
torsion elements.) From now on $V$ will denote a commutative and
cocommutative Hopf algebra without torsion elements,
$V=\mathbb{C}[G]\ten \mathcal{U}(L)$ with $G$ a finitely
generated free Abelian group, and $L$ with countable basis.

This means that we assume  $V$ has the form
\[
V=\mathbb{C}[e^{\pm\alpha_{1}},\dots,e^{\pm\alpha_{l}},x_{1},x_{2},\dots].
\]
We also keep the condition that $r(e^{\alpha_{i}}\ten
e^{\alpha_{j}})=e^{a_{ij}}\in\mathbb{C}$.

Then we will argue that we can find (in the torsion free case) a
quadratic differential operator $\mathbf{Q}$ serving as a logarithm
to the map $\EQ_r$ and establish a relation similar to that of Theorem
\ref{thm:e^QisHomom} for more general Hopf algebras with grouplike elements (as
above).

\begin{defn}
  Let $V$ be a torsion free commutative cocommutative Hopf algebra as above,
  $V=\mathbb{C}(G)\ten \mathbb{C}[x_{1},x_{2},\dots]$,
  $G=\mathbb{Z}[\alpha_1, \alpha_2, \dots , \alpha_k]$. Define
  derivations $\frac{\partial}{\partial \alpha_i}: V\to V$, $i=1, \dots
  , k$, by
\begin{equation}
  \frac{\partial}{\partial \alpha_i} (e^{\alpha }P(x))=m_i  e^{\alpha }P(x)
\end{equation}
for any  $\alpha =\sum _{i=1}^k m_i \alpha_i, \ \ m_i\in \mathbb{Z}$, \  $P(x)\in \mathbb{C}[x_{1},x_{2},\dots]$.
\end{defn}
\begin{example}
We have
\begin{equation*}
  \frac{\partial}{\partial \alpha _i} (e^{\alpha _j}) =\delta _{i j} e^{\alpha _j}.
\end{equation*}
Recall that on $V=\mathbb{C}(G)\ten \mathbb{C}[x_{1},x_{2},\dots]$ we
also have the partial derivatives with respect to the variables $x_n$,
with
\begin{align*}
  &\frac{\partial}{\partial x_n} (e^{\alpha }P(x)) =e^{\alpha
  }(\frac{\partial}{\partial x_n}P(x)), \ \ \text{for \ any }\ n\ge 1,
  \ \ \alpha \ \ \text{as above}.
\end{align*}
Of course, $\frac{\partial}{\partial x_n}$ commutes with $\frac{\partial}{\partial \alpha_{i}}$.
\end{example}

\begin{lem}
  \label{lem:eQis EQ2}
  Let
  $V=\mathbb{C}[e^{\pm\alpha_{1}},\dots,e^{\pm\alpha_{\ell}},x_{1},x_{2},\dots]$
  as before. Let $A$ be a commutative $\mathbb{C}$-algebra and $r$ an
  $A$-valued bicharacter with values on generators
  \begin{equation}\label{eq:4}
\begin{aligned}
&r(e^{\alpha _i}\ten e^{\alpha _j})=e^{a_{i j}}\in \mathbb{C}\\
&r(e^{\alpha _i}\ten x_m)=b_{i m}\in A\\
&r(x_m \ten e^{\alpha _i})=c_{m i}\in A\\
&r(x_i\ten x_m) =q_{mn}\in A
\end{aligned}
\end{equation}
Then the map $\EQ\colon V \to V_{A}, \ \ a\mapsto r(a^{\prime}\ten
a^{\prime\prime})a^{\prime\prime\prime}$ is the exponential of an infinite
  order quadratic operator $\mathbf{Q}$
  \begin{equation}
  \EQ_r(m)=e^{\mathbf{Q}}(m), \ \ m\in V,
  \end{equation}
where $\mathbf{Q}$ is defined by
\begin{multline}\label{eq:6}
  \mathbf{Q}=\sum _{i, j=1}^{\ell} a_{i j} \frac{\partial}{\partial
    \alpha_i}\frac{\partial}{\partial \alpha_j}+\sum
  _{i=1}^{\ell}\sum_{m\ge 1}b_{im} \frac{\partial}{\partial \alpha_i}
  \frac{\partial}{\partial x_{m}}+ \\+\sum_{m\ge 1} \sum
  _{i=1}^{\ell}c_{mi}\frac{\partial}{\partial x_{m}}
  \frac{\partial}{\partial \alpha_i} +
 \sum_{m,n\ge1}q_{mn}\frac{\partial^{2}}{\partial x_{m}\partial
    x_{n}}.
\end{multline}
\end{lem}

\begin{remark}
\label{remark:reasonForLog}
The first equation of ~\eqref{eq:4} is the reason that we require
\\ $r(e^{\alpha _i}\ten e^{\alpha _j})$ to be a constant, as opposed to no such requirement
for instance on the bicharacter $r(x_i\ten x_m)$, which is allowed to be
in a more general target algebra $A$. For example, when
$A=\mathbf{C}[z]$, if $r(e^{\alpha _i}\ten e^{\alpha _j})\in A$ is not
a constant, the "logarithm" $a_{i j}$ is not an element of $A$.
\end{remark}
Before giving the proof, we will start  with some examples.
\begin{example}
\label{example:grouplike}
In particular, for an element $e^{\alpha _i}$  we have:
  \begin{equation}
  e^{\mathbf{Q}}(e^{\alpha _i})=e^{\alpha _i}e^{a_{i i}}=e^{\alpha _i}r(e^{\alpha _i}\ten e^{\alpha _i})=\EQ(e^{\alpha _i})
    \end{equation}
For two independent grouplike  elements $e^{\alpha _i}, \ e^{\alpha _j}$ we have:
 \begin{align*}
   &e^{\mathbf{Q}}(e^{\alpha _i}e^{\alpha _j})=e^{\alpha _i}e^{\alpha _j} e^{a_{i i} +a_{i j} +a_{j i} +a_{j j}}=\\
   & =e^{\alpha _i}e^{\alpha _j}r(e^{\alpha _i}\ten e^{\alpha
     _i})r(e^{\alpha _i}\ten
   e^{\alpha _j})r(e^{\alpha _j}\ten e^{\alpha _i})r(e^{\alpha _j}\ten e^{\alpha _j})=\\
   & =e^{\alpha _i}e^{\alpha _j}r(e^{\alpha _i}e^{\alpha _j}\ten e^{\alpha _i}e^{\alpha _j})=\EQ(e^{\alpha _i}e^{\alpha _j}).\qquad
    \end{align*}
\end{example}
We now  proceed with the proof of the lemma.
\begin{proof}
We can spit the quadratic operator $\mathbf{Q}$ in 3 parts:
\begin{equation}
 \mathbf{Q}=Q_0 +Q_1 +Q_p,
 \end{equation}
 where
 \begin{align}
   &Q_0=\sum _{i, j=1}^{k} a_{i j} \frac{\partial}{\partial \alpha_i}\frac{\partial}{\partial \alpha_j}\\
   &Q_1=\sum _{i=1}^{k}\sum_{m\ge 1} b_{im} \frac{\partial}{\partial
     \alpha_i} \frac{\partial}
   {\partial x_{m}}+ \sum_{m\ge 1} \sum _{i=1}^{k}c_{mi}
   \frac{\partial}{\partial x_{m}}
 \frac{\partial}{\partial \alpha_i}\\
   &Q_p=\sum_{m,n\ge1}q_{mn}\frac{\partial^{2}}{\partial
     x_{m}\partial x_{n}}
   \end{align}
   Notice that the third part, the operator $Q_p$, coincides with the
   operator $Q_p$ used in section \ref{sec:PolyAlg}.  These three
   differential operators commute among each other, thus we can write
 \begin{equation}
 e^{\mathbf{Q}}=e^{Q_1}e^{Q_{p}}e^{Q_0}
 \end{equation}
 Moreover, if $e^{\alpha}, \ \alpha =\sum _{i=1}^k m_i \alpha_i, \
 m_i\in \mathbb{Z}$ is a grouplike element, and $P(x)\in
 \mathbb{C}[x_{1},x_{2},\dots]$, we have
 \begin{equation}
 e^{Q_{p}}(e^{\alpha})=e^{\alpha}, \quad e^{Q_{p}}(e^{\alpha}P(x))=e^{\alpha}e^{Q_{p}}(P(x))
 \end{equation}
 On the other hand, we have
 \begin{equation}
 e^{Q_0}(P(x))=P(x), \quad e^{Q_0}(e^{\alpha}P(x))=e^{Q_0}(e^{\alpha})P(x)
 \end{equation}
 Thus
 \begin{equation}
 e^{\mathbf{Q}}(e^{\alpha}P(x))=e^{Q_1}e^{Q_{p}}e^{Q_0}(e^{\alpha}P(x))=e^{Q_1}
 \big(e^{Q_0}(e^{\alpha})e^{Q_{p}}(P(x))\big)
 \end{equation}
 We know from Theorem \ref{thm:e^QisHomom} that $e^{Q_{p}}(P(x))$ is
 the $\bullet$ polynomial $P^\bullet (x)$, which equals $\EQ(P(x))$.
 It is not hard to show, similar to example \ref{example:grouplike},
 that for purely grouplike elements
 \begin{equation}
 e^{Q_0}(e^{\alpha})=\EQ(e^{\alpha}).
 \end{equation}
 Thus we have
 \begin{equation}
 e^{\mathbf{Q}}(e^{\alpha}P(x))=e^{Q_1}\big(\EQ(e^{\alpha})\EQ(P(x))\big).
 \end{equation}
 We are preparing to use Lemma \ref{lem:EQisHomom}, and to do that we
 need to show that for a grouplike element $e^{\alpha}$, and any polynomial $P(x)\in
 \mathbb{C}[x_{1},x_{2},\dots]$ we have
 \begin{equation}
 e^{Q_1}(e^{\alpha}P(x))=e^{\alpha}\bullet  P(x).
 \end{equation}
 To that end,
 \begin{align*}
   &Q_1(e^{\alpha}P(x))=\big(\sum _{i=1}^{k}\sum_{m\ge 1} b_{im} \frac{\partial}{\partial \alpha_i} \frac{\partial}{\partial x_{m}}+ \sum_{m\ge 1} \sum _{i=1}^{k}c_{mi} \frac{\partial}{\partial x_{m}} \frac{\partial}{\partial \alpha_i}\big)(e^{\alpha}P(x))=\\
  &=e^{\alpha}\big(\sum _{i=1}^{k}m_i(\sum_{m\ge 1}
   b_{im}\frac{\partial}{\partial x_{m}}+ \sum_{m\ge 1} c_{mi}
   \frac{\partial}{\partial x_{m}})\big)(P(x)),
 \end{align*}
 and so
 \begin{align*}
   &e^{Q_1}(e^{\alpha}P(x))=
   e^{\sum _{i=1}^{k} m_i\alpha_i}e^{\big(\sum _{i=1}^{k}m_i(\sum_{m\ge
       1} b_{im}\frac{\partial}{\partial x_{m}}+
     \sum_{m\ge 1} c_{mi} \frac{\partial}{\partial x_{m}})\big)}(P(x))=\\
   &=\prod_{i=1}^{k}\big((e^{\alpha _i})^{m_i}(e^{\sum_{m\ge 1}
     b_{im}\frac{\partial}{\partial x_{m}}})^{m_i}(e^{\sum_{m\ge 1}
     c_{m i}\frac{\partial}{\partial x_{m}}})^{m_i}P(x),
 \end{align*}
 which, using \eqref{eq:PoloCoprod}, is precisely the  $\bullet$ product
 \begin{equation*}
 e^{\alpha}(P(x))^{\prime} s(e^{\alpha}\ten (P(x))^{\prime\prime}),
 \end{equation*}
 where $s$ is the symmetrization bicharacter
 \[
 s(e^{\alpha}\ten P(x))=r(e^{\alpha}\ten (P(x))^{\prime})r((P(x))^{\prime\prime}\ten e^{\alpha}).
\]
 Thus we have according to Lemma \ref{lem:EQisHomom}
 \begin{equation*}
   e^{\mathbf{Q}}(e^{\alpha}P(x))
   =e^{Q_1}\big(\EQ(e^{\alpha})\EQ(P(x))\big)=\EQ(e^{\alpha})\bullet \EQ(P(x))=\EQ(e^{\alpha}P(x)).
 \end{equation*}
\end{proof}
In the paragraph before Theorem \ref{thm:e^QisHomom} we introduced the
notion of a $\bullet$ polynomial $P^{\bullet}$ corresponding to $P\in
V_{0}$. In a similar way we define the $\bullet$ element $a^{\bullet}$
for $a\in V$ as follows. We write $a=e^{\alpha}P$, and put
\[
a^{\bullet}=(e^{\alpha})^{\bullet}\bullet P^{\bullet},
\]
where $(e^{\alpha})^{\bullet}=e^{\alpha}r(e^{\alpha}\ten
e^{\alpha})$. Here we have fixed a bicharacter $r$ with symmetrization $s$,
and the corresponding product $\bullet=\bullet_{s}$.

Thus we obtain the following generalization of Theorem
\ref{thm:e^QisHomom}:

\begin{thm}\label{Thm:GroupLike_EQ_PRop}
  Let $V=\mathbb{C}(G)\ten \mathbb{C}[x_{1},x_{2},\dots]$,
  $G=\mathbb{Z}[\alpha_1, \alpha_2, \dots , \alpha_k]$, with
  bicharacter $r$, taking complex values on grouplike elements, with
  symmetrization $s=r\circ r^{t}$ and associated $\bullet=\bullet_{s}$.  Let
  $\mathbf{Q}$ be the quadratic operator \eqref{eq:6}. Then for $a\in
  V$ we have
\[
e^{\mathbf{Q}}(a)=\EQ_{r}(a)=r(a^{\prime}\ten
a^{\prime\prime})a^{\prime\prime\prime}.
\]
Moreover
\[
e^{\mathbf{Q}}(a)=a^{\bullet},
\]
and $e^{\mathbf{Q}}$ is a homomorphism from $(V\ten A,\cdot)\to(V\ten
A,\bullet)$.
\end{thm}
\begin{example}
In particular, from an element $e^{\alpha _i}$  in $V$ we obtain $(e^{\alpha _i})^{\bullet}$:
\begin{equation}
  (e^{\alpha _i})^{\bullet} =e^{\alpha _i}r(e^{\alpha _i}\ten e^{\alpha _i}) =e^{\alpha _i}e^{a_{i i}}=e^{\mathbf{Q}}(e^{\alpha _i})
\end{equation}
For two independent grouplike elements $e^{\alpha _i}, \ e^{\alpha_j}$ we have:
 \begin{align*}
   &(e^{\alpha _i}e^{\alpha _j})^{\bullet}
   =\EQ(e^{\alpha_i}e^{\alpha_j})
   =e^{\alpha _i}e^{\alpha _j}r(e^{\alpha_i}e^{\alpha _j}\ten e^{\alpha _i}e^{\alpha _j})=
   \\
   &=(e^{\alpha_i})^{\bullet} (e^{\alpha_j})^{\bullet} s(e^{\alpha_i}\ten
   e^{\alpha _j})
   =(e^{\alpha_i})^{\bullet} \bullet (e^{\alpha_j})^{\bullet}
   = e^{\alpha_i}e^{\alpha_j} e^{a_{i i} +a_{i j}
     +a_{j i} +a_{jj}}=
   \\
   &=e^{\mathbf{Q}}(e^{\alpha_i}e^{\alpha_j}).
    \end{align*}
\end{example}
\begin{remark}
  \label{rem:torsion_elements}
We now discuss briefly the effect of torsion elements on the above
results. Let $M$ be a Hopf algebra with torsion elements:
\[
M=\mathbb{C}[G_{\mathrm{Tor}}]\ten
V=\mathbb{C}[\delta_{1},\delta_{2},\dots,\delta_{s},
e^{\pm\alpha_{1}},\dots,e^{\pm\alpha_{l}},x_{1},x_{2},\dots],
\]
with $\delta_{i}$ torsion elements. An element $a$ of $M$ is then of
the form
\[
a=\delta e^{\alpha}P(x).
\]
We fix a symmetric bicharacter on $M$ (taking complex values on
grouplikes), and let $r$ be some bicharacter with $s$ as
symmetrization. We define the $\bullet$ element of $a$ as
\[
a^{\bullet}=(\delta^{\bullet})\bullet (e^{\alpha})^{\bullet}\bullet
P^{\bullet},
\]
where $\delta^{\bullet}=\delta r(\delta \ten\delta)$. (So
$\delta^{\bullet}$ differs from $\delta$ by a root of unity.)

We have then, just as before,
\[
\EQ_{r}(a)=a^{\bullet}, \quad
\EQ_{r}(ab)=(a^{\bullet})\bullet(b^{\bullet}).
\]
If we want to write $\EQ$ in terms of a quadratic differential
operator we have, for $a=\delta e^{\alpha}P(x)$
\[
\EQ_{r}(a)=(\delta^{\bullet})\bullet e^{Q_{V}}(e^{\alpha}P(x)).
\]
Here $\mathbf{Q}_{V}$ is the differential operator \eqref{eq:6}
constructed from the restriction of the bicharacter $r$ defined on
$M=\mathbb{C}[G_{\mathrm{Tor}}]\ten V$ to $V$.
\end{remark}


\section{The Frenkel-Lepowsky-Meurman Example}
\label{sec:The Frenkel-Lepowsky-Meurman Example}

For the reader's convenience we explain in this last section why the
operator $e^{\Delta_{z}}$ of \cite{MR996026}, which plays a prominent
role in the construction of twisted modules over a lattice vertex
algebra, is a special case of the operator we call $e^{\mathbf{Q}}$.

Recall the torsion free commutative and cocommutative Hopf algebra
$V$, where
\[
V=\mathbb{C}[G]\ten V_{0}=\mathbb{C}[G]\ten
\mathcal{U}(L)=\mathbb{C}[e^{\pm\alpha_{1}}.\dots,
e^{\pm\alpha_{\ell}}]\ten \mathbb{C}[x_{1},\dots, x_n, \dots ],
\]
studied in Section \ref{sec:NonTrivGrpLike}. Until now the Abelian
group $G$ and the Abelian Lie algebra $L$ were independent. In the
application to vertex algebras this is no longer true: $L$ is in fact
constructed from the group $G$, by use of an  extra structure.

One starts with a lattice $Q$, i.e., a free Abelian group
$Q=\oplus_{i=1}^{\ell}\mathbb{Z}\alpha_{i}$ equipped with a symmetric bilinear
form $(\alpha,\beta)\mapsto \langle
\alpha\mid\beta\rangle\in\mathbb{C}$. We will assume that the bilinear
form is nondegenerate. Then, in order to construct $L$ (and hence
$V_{0}$), we complexify the lattice: Define
\[
\mathfrak h=\mathbb{C}\ten_{\mathbb{Z}} Q.
\]
Choose an orthonormal basis $h^{s}, s=1,2,\dots,\ell$ for
$\mathfrak{h}$,
\[
\langle h^{s}\mid h^{t}\rangle=\delta_{st}.
\]
Then we let $L$ be the Abelian Lie algebra with basis $h^{s}(-n)$,
$1\le s\le\ell$, $n\in \mathbb{N}$, and
$V_{0}=\mathcal{U}(L)=\mathbb{C}[h^{s}(-n)]$. (So $V_{0}$ has now the
$h^{s}(-n)$,
$1\le s\le\ell$, $n\in \mathbb{N}$ as generating countable set of primitive elements,
instead of the $x_{n}$ as we used before.) We then let
$V=\mathbb{C}[Q]\ten V_{0}$.

On $V$ we have operators $\frac{\partial}{\partial h^{s}(-n)}$ and
$\frac{\partial}{\partial \alpha_{i}}$, the ingredients of the
quadratic differential operator \eqref{eq:6}. In the vertex algebra
literature it is usual to introduce notations
\begin{equation}
h^{s}(n)=n \frac{\partial}{\partial h^{s}(-n)},\quad 1\le s\le\ell,
n>0, \ n\in \mathbb{N}, \label{eq:10}
\end{equation}
and for $n=0$
\begin{equation}
h^{s}(0)=\sum_{i=0}^{\ell}\langle h^{s}\mid \alpha_{i}\rangle
\frac{\partial}{\partial \alpha_{i}}.\label{eq:11}
\end{equation}
Clearly $\{h^{s}(0)\}$ is another basis of the space of derivations of
$V$ spanned by the $\frac{\partial}{\partial \alpha_{i}}$ (here we use the assumption that the bilinear form is nondegenerate). Hence we
can compactly write an alternative "Heisenberg" form (see Appendix \ref{sec:NormalOrdProd} for the reason for this name) of a  quadratic differential operator \eqref{eq:6} as
\begin{equation}
\mathbf{Q}=\sum_{s,t=1}^{\ell}\sum_{m,n=0}^{\infty}c_{mn}^{st}h^{s}(m)h^{t}(n),\label{eq:7}
\end{equation}
where the coefficients $c_{mn}^{st}$ are expressed (invertibly) in terms of values
of some bicharacter as in \eqref{eq:4},
$r(e^{\alpha_{i}}\ten e^{\alpha_{j}}), r(e^{\alpha_{i}}\ten
h^{s}(-n)),$ etc. We will give explicit formulas in a special case below.

So it remains to choose a bicharacter, or, equivalently, to choose the
constants $c_{mn}^{st}$. Particularly nice formulas arise when the
constants $c_{mn}^{st}$ are independent of $s$ and $t$ (as is the case in  \cite{MR996026}), so that we
obtain
\begin{equation}
Q=\Delta_{z}=\sum_{s,t=1}^{\ell}\sum_{m,n=0}^{\infty}c_{mn}h^{s}(m)h^{t}(n),\label{eq:8}
\end{equation}
and the quadratic differential operator of interest in the theory of
twisted modules is specified by choosing a generating series for the
$c_{mn}$ as
\begin{equation}
\sum_{m,n=0}^{\infty}c_{mn}x^{m}y^{n}=-\log\frac{\sqrt{1+\frac{x}{z}}+\sqrt{1+\frac{y}{z}}}{2}.
\label{eq:12}
\end{equation}
Here one expands the right-hand-side as a Maclaurin power series in the variables $x$ and $y$, treating $z$ as a parameter. Note that in this example the choice of the target algebra for the bicharacter $r: V\ten V\to A$ is $A=\mathbb{C}[\frac{1}{z}]$, hence the notation $\Delta_{z}$.

Of course, the substitutions \eqref{eq:10}, \eqref{eq:11} and the
choice of generating series \eqref{eq:12} are not well motivated from
the Hopf algebraic point of view. See the original \cite{MR996026} for
the vertex algebraic context, and \cite{MR2681778} for another
approach.

If one wants to have explicit formulas for the bicharacter that
corresponds to the quadratic differential operator  we need to compare the two forms of the operator $\mathbf{Q}$: the \eqref{eq:8} and  \eqref{eq:6}. We are
now using as primitive elements, instead of $x_{n}$, the elements
$h^{s}(-n)$, so that instead of the form \eqref{eq:6} we get  the expression
\begin{align}
  \label{eq:20}
  \mathbf{Q}&=\sum _{i, j=1}^{\ell} a_{i j} \frac{\partial}{\partial
    \alpha_i}\frac{\partial}{\partial \alpha_j}
  +\sum_{i=1}^{\ell}\sum_{m\ge 1} b_{im}^j \frac{\partial}{\partial \alpha_i}
  \frac{\partial}{\partial h^j(-m)} +\\ \notag
  &\quad + \sum_{m\ge 1} \sum
  _{i=1}^{\ell}c_{mi}^j \frac{\partial}{\partial h^j(-m)}
  \frac{\partial}{\partial \alpha_i}
  +\sum _{i,j=1}^{\ell}\sum_{m,n\ge1}q_{mn}^{i j}\frac{\partial^{2}}{\partial
    h^i(-m)\partial h^{j}(-n)},
\end{align}
By comparing the two forms of $\mathbf{Q}$, \eqref{eq:8} and  \eqref{eq:20}, and using the orthonormality of the basis $h^{i}$, we have
\begin{align*}
& a_{i j}=c_{00}\langle \alpha_{i}  | \alpha_{j} \rangle, \quad \text{for any} \quad i, j=1, \dots k\\
& mb_{im}^j= c_{0 m} \sum _{s=1}^k  \langle  h^{s} | \alpha_{i}  \rangle, \quad \text{for any} \quad j=1, \dots k,\quad m\in \mathbb{N},\\
& mc_{mi}^j= c_{m 0} \sum _{s=1}^k  \langle  \alpha_{i} | h^{s}  \rangle, \quad \text{for any} \quad j=1, \dots k \quad m\in \mathbb{N}.
\end{align*}
Thus the corresponding bicharacter on the grouplike elements from Lemma \ref{lem:eQis EQ2}  simplifies to:
\begin{equation}
\label{eq:GroGroLat}
r(e^{\alpha }\ten e^{\beta})=(e^{c_{00}})^{\langle \alpha  | \beta  \rangle}.
\end{equation}
In order to find similar "lattice" formulas for the rest of the
bicharacters from Lemma \ref{lem:eQis EQ2}, we turn to a basis
typically used in the theory of vertex algebras. We consider the
following degree 1 elements  in the polynomial algebra ${V}_0$. For each $\alpha_i \in
L$, let
\begin{equation}
  \alpha_i(-m) =m\sum_{s=1}^{\ell}\langle h^{s}  | \alpha_i  \rangle
  h^s(-m),
  \quad i=1, \dots , k; \quad m\in \mathbb{N},
\end{equation}
\begin{lem}
\label{lem:Natbasis}
The  elements $\alpha_i(-m)$ are primitive for any $i=1, \dots , k; \ m\in \mathbb{N}$. For these primitive elements  the bicharacters from
Lemma \ref{lem:eQis EQ2} assume the form:
\begin{align}
\label{eq:GroPrimLat1}
&r(e^{\alpha}\ten \alpha_{i}(-m))=\langle \alpha  | \alpha _i  \rangle c_{0 m}\\
\label{eq:GroPrimLat2}
&r(\alpha_{i}(-m)\ten e^{\alpha})=\langle \alpha_i | \alpha \rangle c_{m 0}\\
\label{eq:GroPrimLat3}
&r(\alpha_{i}(-m) \ten \alpha_{j}(-n)) = \langle \alpha_i  | \alpha ^j  \rangle c_{mn}
\end{align}
\end{lem}
\begin{proof}
Since $h^s(-m)$ are primitive elements for any $s=1, \dots ,k, \ m\in
\mathbb{N}$, it follows that $\alpha_i(-m)$ are primitive for any
$i=1, \dots ,k, \ m\in \mathbb{N}$.
Using the property
\begin{align}
  r(e^{\alpha_{j}}\ten \alpha_i(-m)) & =\sum_{s=1}^{\ell}\langle h^{s}
  | \alpha_i  \rangle r(e^{\alpha_{j}}\ten m h^s(-m))
  = \sum_{s=1}^{\ell}\langle h^{s}  | \alpha_i  \rangle m b_{jm}^s =\notag\\
\label{eq:Gr-Prim}
& =c_{0 m} \sum_{s=1}^{\ell} \sum _{l=1}^k \langle h^{s}  | \alpha_i
\rangle    \langle  h^{l} | \alpha_{j} \rangle
=c_{0 m}\langle \alpha_{j}  | \alpha_i  \rangle.
\end{align}
 For any primitive element $x$ and
grouplike elements $g_1, g_2$ we have
\begin{equation}
r(g_1g_2\ten x)=r(g_1\ten x) +r(g_2\ten x).
\end{equation}
Hence from \eqref{eq:Gr-Prim} the equality \eqref{eq:GroPrimLat1}
follows immediately. Equalities \eqref{eq:GroPrimLat2} and
\eqref{eq:GroPrimLat3} follow similarly.
\end{proof}

As was mentioned above, I. Frenkel, J. Lepowsky and A. Meurman defined
the quadratic differential operator $\Delta_{z}$ by choosing the
generating function \eqref{eq:12}. As shown above in equation
\eqref{eq:GroGroLat} and lemma \ref{lem:Natbasis}, in the natural
basis $\alpha_i(-m), \ i=1, \dots , k; \ m\in \mathbb{N}$ the
bicharacter $r$ is explicit and simple. Thus we can complete the
picture and define a $\bullet _s$-product, together with the
bicharacter map $\EQ_r$ (the bicharacter $s$ being the symmetrization
of $r$), so that the map $e^{\Delta_{z}}=e^{\mathbf{Q}}=\EQ_r$.

\begin{example}
\label{example:TheBigEx}
Let us finish with calculating some examples of the $\bullet$ product
on $V=\mathbb{C}(L)\ten {V}_0$ in the specific case outlined
above.

We have the generating series \eqref{eq:12} for the coefficients
$c_{mn}$ of the quadratic differential operator $\mathbf{Q}$.
The first few terms of the series expansion are
\begin{equation}
c_{00}+c_{01}y+c_{10}x+c_{11}xy+\dots=-\frac{1}{4z}(x +y) +\frac{3}{32z^2}(x^2 +y^2) +\frac{1}{16z^2}xy +\dots .
\end{equation}
In particular note that $c_{00}=0$ in this case. Thus we have
\begin{align*}
&r(e^{\alpha }\ten e^{\beta})=(e^{c_{00}})^{\langle \alpha  | \beta  \rangle}=1,\\
&r(e^{\alpha }\ten \alpha_{i}(-1))=\langle \alpha  | \alpha _i  \rangle c_{0 1}=-\frac{\langle \alpha  | \alpha _i  \rangle}{4z},\\
&r(\alpha_{i}(-1)\ten e^{\alpha})=\langle \alpha _i | \alpha \rangle c_{1 0}=-\frac{\langle \alpha _i | \alpha   \rangle}{4z},\\
&r(\alpha_{i}(-1) \ten \alpha_{j}(-1)) = \langle \alpha_i  | \alpha_j  \rangle c_{11}=\frac{\langle \alpha_i  | \alpha_j  \rangle}{16z^2}.
\end{align*}
Hence for the symmetrization bicharacters we have
\begin{align*}
  &s(e^{\alpha }\ten e^{\beta})=r(e^{\alpha }\ten e^{\beta})r(e^{\beta }\ten e^{\alpha})=1,\\
  &s(e^{\alpha }\ten \alpha_{i}(-1))=r(e^{\alpha }\ten \alpha_{i}(-1))
  +r(\alpha_{i}(-1)\ten e^{\alpha})
  =-\frac{\langle \alpha  | \alpha _i  \rangle}{2z},\\
  &s(\alpha_{i}(-1) \ten \alpha_{j}(-1)) = r(\alpha_{i}(-1) \ten
  \alpha_{j}(-1)) +r(\alpha_{i}(-1) \ten \alpha_{j}(-1))=\\
&\qquad\qquad\qquad\qquad=\frac{\langle \alpha_i
    | \alpha_{j} \rangle}{8z^2}.
\end{align*}
We have then
\begin{align*}
  e^{\mathbf{Q}}(e^{\alpha_{i}}) & =\EQ_r(e^{\alpha_{i}})=(e^{\alpha_{i}})^\bullet =e^{\alpha_{i}}\\
  e^{\mathbf{Q}}(e^{\alpha_{i}}e^{\alpha_{j}}) & =\EQ_r(e^{\alpha_{i}}e^{\alpha_{j}})=(e^{\alpha_{i}}e^{\alpha_{j}})^\bullet =(e^{\alpha_{i}})^\bullet  \bullet (e^{\alpha_{j}})^\bullet =e^{\alpha_{i}} \bullet e^{\alpha_{j}}= \\
  & =e^{\alpha_{i}}e^{\alpha_{j}} s(e^{\alpha_{i}} \ten e^{\alpha
   _{j}})=e^{\alpha_{i}}e^{\alpha_{j}}.
\end{align*}
Thus in this case we have for any $\alpha =\sum _{i=1}^k m_i \alpha_i, \ m_i\in \mathbb{Z}$
\begin{equation*}
e^{\mathbf{Q}}(e^{\alpha })=\EQ_r(e^{\alpha })=(e^{\alpha })^\bullet =e^{\alpha },
\end{equation*}
Further,
\begin{align*}
  e^{\mathbf{Q}}(\alpha_i(-1))
  & =\EQ_r(\alpha_i(-1))=(\alpha_i(-1))^\bullet =\alpha_i(-1)\\
  e^{\mathbf{Q}}(e^{\alpha }\alpha_i(-1))
  & =\EQ_r(e^{\alpha }\alpha_i(-1))=(e^{\alpha }\alpha_i(-1))^\bullet =(e^{\alpha })^{\bullet }\bullet (\alpha_i(-1))^\bullet = \\
  &=(e^{\alpha })\bullet (\alpha_i(-1)) =e^{\alpha }\alpha_i(-1) +s(e^{\alpha }\ten \alpha
 _i(-1))=\\
 &\qquad =e^{\alpha }\alpha_i(-1)-\frac{\langle \alpha
    | \alpha _i  \rangle}{2z} \\
  e^{\mathbf{Q}}((\alpha_i(-1))^2)& =\EQ_r((\alpha_i(-1))^2)=(\alpha
 _i(-1))^\bullet  \bullet (\alpha_i(-1))^\bullet = \\
  &= (\alpha_i(-1))
  \bullet (\alpha_i(-1))=(\alpha_i(-1))^2 +s(\alpha_i(-1) \ten \alpha_j(-1)) = \\
  & =(\alpha_i(-1))^2 +\frac{\langle \alpha_i | \alpha_j
    \rangle}{8z^2}.
\end{align*}
Even though on lower degree products it is about as easy to
calculate the $\bullet$ product or the action of $e^{\mathbf{Q}}$, on higher
products it is much easier to calculate the $\bullet$ products (of
course, Theorem \ref{Thm:GroupLike_EQ_PRop} assures that we will get
the same result). For instance:
\begin{align*}
e^{\mathbf{Q}}(e^{\alpha }(\alpha_i(-1))^2)=
&(e^{\alpha }(\alpha_i(-1))^2)^\bullet=(e^{\alpha })^{\bullet }\bullet
((\alpha_i(-1))^2)^\bullet =\\
&=e^{\alpha }\bullet \big((\alpha_i(-1))^2+\frac{\langle \alpha_i  | \alpha_j  \rangle}{8z^2}\big)=e^{\alpha }\bullet (\alpha_i(-1))^2 +e^{\alpha }\frac{\langle \alpha_i  | \alpha_j  \rangle}{8z^2}.
\end{align*}
Since $e^{\alpha }$ is grouplike, we have
\begin{equation*}
s(e^{\alpha }\ten (\alpha_i(-1))^2)=(s(e^{\alpha }\ten (\alpha_i(-1))))^2.
\end{equation*}
From
\begin{align*}
\del ((\alpha_i(-1))^2) =(\alpha_i(-1))^2\ten 1 +2\alpha_i(-1)\ten \alpha_i(-1) +1\ten (\alpha_i(-1))^2
\end{align*}
we have
\begin{align*}
  e^{\alpha }\bullet (\alpha_i(-1))^2& =e^{\alpha
  }(\alpha_i(-1))^2 +2e^{\alpha }(\alpha_i(-1))s(e^{\alpha }\ten
  \alpha_i(-1)) +\\
&\qquad\qquad +e^{\alpha }s(e^{\alpha }\ten (\alpha_i(-1))^2)=\\
  & =e^{\alpha }\Big( (\alpha_i(-1))^2 -2\alpha_i(-1)\frac{\langle
    \alpha | \alpha _i \rangle}{2z} +\frac{\langle \alpha | \alpha _i
    \rangle ^2}{4z^2} \Big).
\end{align*}
So
\begin{align*}
  e^{\mathbf{Q}}(e^{\alpha }(\alpha_i(-1))^2)=e^{\alpha }\Big( (\alpha
 _i(-1))^2 -2\alpha_i(-1)\frac{\langle \alpha | \alpha _i
    \rangle}{2z} +\frac{\langle \alpha | \alpha _i \rangle ^2}{4z^2}
  +\frac{\langle \alpha_i | \alpha_j \rangle}{8z^2}\Big)
\end{align*}
The rest of the  $\bullet$ products are similarly obtained from the Taylor expansion.
\end{example}

\appendix

\section{Normal ordered products for the Heisenberg algebra}
\label{sec:NormalOrdProd}

In this Appendix we recall the notion of normal ordered
products of fields in the case of the Heisenberg algebra. The notion
of normal ordered products has long been very common in the physics
literature on conformal field theory, and has been introduced in the
mathematical literature (in greater generality than presented here) by
works like \cite{MR996026}, \cite{MR1021978} and others.

We start with the setup of Section \ref{sec:The
  Frenkel-Lepowsky-Meurman Example}, where we have a lattice $Q$ with
complexification $\mathfrak h$, but with the simplification that
$Q=\mathbb{Z}\alpha$ is rank 1, so that $\mathfrak h$ has dimension
1. We fix a basis element $h\in\mathfrak h$ (say of unit length) and
simplify the notation and write in this case $x_{n}=h(-n)$, so that we
deal with the polynomial algebra
\[
V_{0}=\mathbb{C}[x_{1},x_{2},\dots].
\]
On $V_{0}$ we have \emph{creation operators}
\[
h_{-n}=\text{multiplication by $x_{n}$,}
\]
and \emph{annihilation operators} (cf., \eqref{eq:10})
\[
h_{n}=n\frac{\partial}{\partial x_{n}}.
\]
We extend the action of creation and annihilation operators to all of \\
$V=\mathbb{C}[Q]\ten V_{0}=\oplus_{n\in\mathbb{Z}}
V_{0}e^{n\alpha}$. We also define (cf., \eqref{eq:11})
\[
h_{0}=\frac{\partial}{\partial \alpha}.
\]
Let $\mathcal{H}$ be the infinite dimensional Lie algebra generated by the operators
$h_n$ for $n\in \mathbb{Z}$, and ${c}$--a central element, satisfying the relations
\begin{equation}
[h_m,h_n]=m\delta _{m+n,0}{c}, \quad m,n\in \mathbb{Z}.
\end{equation}
$\mathcal{H}$ is called the Heisenberg algebra.  It is clear that
$V$ is a representation of $\mathcal{H}$, with the central element $c$ acting as  multiplication by $1$.

We can organize the Heisenberg operators from $\mathcal{H}$ in a
formal series, called \emph{Heisenberg field}:
\begin{equation}
h(z)=\sum_{n\in \mathbb{Z}} h_n z^{-n-1}.
\end{equation}
The indexing is due to the fact that we want the annihilation
operators to be indexed by negative powers of the formal variable $z$,
and creation operators to be indexed by non negative powers of $z$:
\begin{equation}
 h(z)=h_{-}(z)+h_{+}(z),\label{eq:9}
\end{equation}
 where
$h_{-}(z)=\sum_{n\ge 0} h_n z^{-n-1}$ is called \emph{annihilation part} of the Heisenberg field and
$h_{+}(z)= \sum_{n\ge 0} h_{-n-1} z^{n}$---\emph{creation part} of the Heisenberg field.

The product of two Heisenberg fields with the same formal variable
does not make sense, even when it acts on the element $1_{V_{0}}\in V_{0}$: If
one naively was to multiply
\begin{equation*}
  "h(z)h(z) =\sum_{m\in \mathbb{Z}} h_m z^{-m-1} \sum_{n\in \mathbb{Z}} h_{-n} z^{n-1}
  = \sum_{k\in \mathbb{Z}} z^{-k-2}(\sum_{m-n =k} h_m h_{-n} )" ,
\end{equation*}
one has infinite sums as coefficients, for example, the coefficient in
front of $z^{-2}$ for $h(z)h(z) 1_{V_{0}}$ is $\sum_ {m>0} m$.  To
rectify this, following physicists, one introduces the notion of
normal ordered products:
\begin{defn} (\cite{MR996026}, \cite{MR1021978}, \cite{MR1651389}) \textbf{(Normal ordered products)} \\
  First, let\label{defn:normalorderedproduct}
\begin{align*}
  &\colon h_nh_m \colon =h_n h_m \quad \text{if} \ \ m<0 \\
  &\colon h_nh_m \colon =h_m h_n \quad \text{if} \ \ m\ge 0.
\end{align*}
Then define
\begin{equation}\label{eq:14}
\colon h(z) h(z) \colon =\sum_{k\in \mathbb{Z}} z^{-k-2}(\sum_{m-n =k} \colon h_m h_{-n}\colon ),
\end{equation}
called \textbf{normal ordered products of fields}.
\end{defn}

The normal ordered product of the Heisenberg fields has a well defined
action on any element of $V$.

Similarly to the $\bullet$ products, see Example
\ref{Example:succesivebulletproduct}, one can define normal ordered
products of arbitrary number of fields by a consecutive application
from the right:
\begin{equation}\label{eq:13}
\colon h(z) h(z) h(z) \colon =\colon h(z) \colon h(z) h(z) ::.
\end{equation}
Besides the Heisenberg field $h(z)$ one  also considers derivative fields
$\partial^{i}h(z)$, where $\partial=\partial_z$, and define similarly normal ordered product of those fields.

Now we are ready to define the \textbf{vertex algebra state-field
  correspondence}, which is a map from $V_{0}$ to the space of fields
on $V_{0}$ (or on $V$). It is given by
\begin{equation}
  x_{n_{1}}x_{n_{2}}\dots x_{n_{k}} =h_{-n_{1}}h_{-n_{2}}\dots
  h_{-n_{k}} 1_{V_{0}} \mapsto
\colon \frac{\partial ^{n_{1}-1}h(z)}{(n_{1}-1)!} \frac{\partial
  ^{n_{2}-1}h(z)}{(n_{2}-1)!}\dots
\frac{\partial ^{n_{k}-1}h(z)}{(n_{k}-1)!}\colon
\end{equation}
This is in fact a one-to-one map, the inverse being given as following:
\begin{fact}(\cite{MR996026}, \cite{MR1651389})\textbf{(Field-state correspondence)} \\
\label{fact:field-state}
To a normal product of Heisenberg fields one associates back the
product of states given by:
\begin{align*}
  &\colon \frac{\partial ^{n_{1}-1}h(z)}{(n_{1}-1)!} \frac{\partial
    ^{n_{2}-1}h(z)}{(n_{2}-1)!}\dots
\frac{\partial ^{n_{k}-1}h(z)}{(n_{k}-1)!}\colon
\mapsto  \\
&\colon \frac{\partial ^{n_{1}-1}h(z)}{(n_{1}-1)!} \frac{\partial
    ^{n_{2}-1}h(z)}{(n_{2}-1)!}\dots
\frac{\partial ^{n_{k}-1}h(z)}{(n_{k}-1)!}\colon 1_{V_{0}}\arrowvert_{z=0} =
h_{-n_{1}}h_{-n_{2}}\dots h_{-n_{k}}1_{V_{0}} =\\
 & =x_{n_{1}}x_{n_{1}}\dots x_{n_{k}}.
\end{align*}
\end{fact}

The fact that the evaluation of the normal product of the fields at
$z=0$ makes sense, and gives back the product of the states, is proved
in many books, see for example \cite{MR996026}, \cite{MR1651389},
\cite{MR2023933}.

Similar to the Heisenberg algebra there is the twisted Heisenberg
algebra:
\begin{defn}(\cite{MR996026}, \cite{MR1272580}, \cite{MR2172171})
  \textbf{(Twisted Heisenberg algebra)}
  Let $\mathcal{H}_{\sfrac{1}{2}}$ be the infinite dimensional Lie
  algebra generated by the operators $h_n$ for $n\in \mathbb{Z}
  +\sfrac{1}{2}$, and $\tilde{c}$--a central element, satisfying the
  relations
\begin{equation}
[h_m,h_n]=m\delta _{m+n,0}\tilde{c}, \quad m,n\in \mathbb{Z}+\sfrac{1}{2}.
\end{equation}
The generators are organized in the \textbf{twisted Heisenberg field}:
\begin{equation}
\tilde{h}(z)=\sum_{n\in \mathbb{Z}+\sfrac{1}{2}} h_n z^{-n-1}.
\end{equation}
\end{defn}

Just as the Heisenberg algebra acts on $V_{0}$ one can define also a
module, say $\tilde V_{0}$, for the twisted Heisenberg algebra, such that the annihilation
operators are still $h_{n}, n>0$ and creation operators
$h_{n},n<0$. Since we have the notion of creation and annihilation
operators we can define normal ordered products of (derivatives of)
the twisted Heisenberg fields. Then we can also define a state-field
correspondence, from states in the (untwisted) space $V_{0}$ to the
twisted fields that act on $\tilde V_{0}$.

One of the questions  we  started in this paper is the
following:\\
\textbf{For a twisted Heisenberg algebra, normal ordered products of
  twisted fields correspond to what products of states (under the
  twisted state field correspondence)? }
\\ I.e., what is the
equivalent of the field-state correspondence \ref{fact:field-state} on
the module $\tilde V_{0}$ for the twisted Heisenberg algebra
$\mathcal{H}_{\sfrac{1}{2}}$? It is obvious one can no longer apply
the evaluation at $z=0$ as in \ref{fact:field-state}.

Nevertheless, we can now formulate an answer to this question, and we leave the proof to the reader familiar with twisted modules of vertex algebras: \\
Let the bicharacter $r$ and its symmetrization $s$ be defined as in
Section \ref{sec:The Frenkel-Lepowsky-Meurman Example}, see Example
\ref{example:TheBigEx}, and consider their inverses as in
\eqref{eq:3}; then
\begin{align*} &\colon \frac{\partial
^{n_{1}-1}\tilde{h}(z)}{(n_{1}-1)!}  \frac{\partial
^{n_{2}-1}\tilde{h}(z)}{(n_{2}-1)!}\dots \frac{\partial
^{n_{k}-1}\tilde{h}(z)}{(n_{k}-1)!}\colon \mapsto
e^{-Q_p}(x_{n_{1}}x_{n_{2}}\dots x_{n_{k}}) =\\ &=\EQ_{r^{-1}}
(x_{n_{1}}x_{n_{2}}\dots x_{n_{k}}) = x_{n_{1}}\bullet
_{s^{-1}} x_{n_{2}}\bullet _{s^{-1}} \dots \bullet _{s^{-1}}
x_{n_{k}}.
\end{align*}
So the normal ordered product of twisted fields corresponds to
the $\bullet$ product on $V_{0}$.

\section{Operator description of the coproduct}
\label{sec:alternativeproof}

We want to introduce an alternative operator description of the coproduct involving grouplike elements, similar to the well known description we used in Equation \eqref{eq:PoloCoprod}. To do that we need expressions involving
  exponentials of $\alpha\in G$ and $\frac{\partial}{\partial
    \alpha_{i}}$. Recall that in Section \ref{sec:PolyAlg} we used expressions
$e^{x_{n}\frac{\partial}{\partial_{x_{n}}}}$ for primitive elements $x_n$; such an expression was
interpreted as a power series $\sum
\frac{1}{s!}(x_{n}\frac{\partial}{\partial_{x_{n}}})^{s}$. (It is a
locally finite infinite order differential operator on
$\mathbb{C}[x_{1},\dots]$.) Now we want to consider the  expression
$e^{\alpha_{i}\frac{\partial}{\partial_{\alpha_{i}}}}$ as an operator  on
$\mathbb{C}[G]$. This can not be interpreted as a power series, as the
powers of $\alpha_{i}$ do not belong to $\mathbb{C}[G]$. However,
$\frac{\partial}{\partial_{\alpha_{i}}}$ is diagonalizable on
$\mathbb{C}[G]$ (and on $\mathbb{C}[G]\ten
\mathbb{C}[x_{1},\dots]$). So on the eigenspace for
$\frac{\partial}{\partial_{\alpha_{i}}}$ with eigenvalue $m_{i}$ the
exponential operator
$e^{\alpha_{i}\frac{\partial}{\partial_{\alpha_{i}}}}$ is just
multiplication by $e^{m_{i}\alpha_{i}}=(e^{\alpha_{i}})^{m_{i}}$.

As in section \ref{sec:PolyAlg} we can use the operators $e^{\alpha_{i}\frac{\partial}{\partial_{\alpha_{i}}}}$ to give a convenient description of the
coproduct of $V$. An element of $V$ is a linear combination of elements of  the form $e^{\alpha}P(x)$,
$\alpha=\sum m_{i}\alpha_{i}$, $P(x)\in V_{0}$. The coproduct of $e^{\alpha}P(x)$ is
\begin{align*}
  \Delta(e^{\alpha}P(x))&=e^{\alpha}\ten e^{\alpha}P(x^{(1)}+
  x^{(2)})=\\
  &=e^{\alpha^{(1)}+\alpha^{(2)}}P(x^{(1)}+x^{(2)})=\\
  &=e^{\sum_i\alpha_{i}^{(1)}\frac{\partial}{\partial{\alpha_{i}^{(2)}}}
    +\sum_{n}x_{n}^{(1)}\frac{\partial}{\partial
      x_{n}^{(2)} }}[e^{\alpha^{(2)}}P(x^{(2)})].
\end{align*}
Here we write $e^{\alpha^{(1)}}$ for $e^{\alpha}\ten 1$, etc. In the
same way for the square of the coproduct:
\begin{equation}
  \Delta^{2}(e^{\alpha}P(x))=
  e^{\sum_i[\alpha_{i}^{(1)}+\alpha_{i}^{(2)}]\frac{\partial}{\partial{\alpha_{i}^{(3)}}}
    +\sum_{n}[x_{n}^{(1)}+x_{n}^{(2)}]\frac{\partial}{\partial
      x_{n}^{(3)} }}[e^{\alpha^{(3)}}P(x^{(3)})].
\label{eq:5B}
\end{equation}
Now we want to show how we can use this alternative operator description of the coproduct to give an alternative proof of Lemma \ref{lem:eQis EQ2}: We recall it for convenience:
\begin{lem}
  Let
  $V=\mathbb{C}[e^{\pm\alpha_{1}},\dots,e^{\pm\alpha_{\ell}},x_{1},x_{2},\dots]$
  as before. Let $A$ be a commutative $\mathbb{C}$-algebra and $r$ an
  $A$-valued bicharacter with values on generators
  \begin{equation}\label{eq:4B}
\begin{aligned}
&r(e^{\alpha _i}\ten e^{\alpha _j})=e^{a_{i j}}\in \mathbb{C}\\
&r(e^{\alpha _i}\ten x_m)=b_{i m}\in A\\
&r(x_m \ten e^{\alpha _i})=c_{m i}\in A\\
&r(x_i\ten x_m) =q_{mn}\in A
\end{aligned}
\end{equation}
Then the map $\EQ\colon V \to V_{A}, \ \ a\mapsto r(a^{\prime}\ten
a^{\prime\prime})a^{\prime\prime\prime}$ is the exponential of an infinite
  order quadratic operator $\mathbf{Q}$
  \begin{equation}
  \EQ_r(m)=e^{\mathbf{Q}}(m), \ \ m\in V,
  \end{equation}
where $\mathbf{Q}$ is defined by
\begin{multline}\label{eq:6B}
  \mathbf{Q}=\sum _{i, j=1}^{\ell} a_{i j} \frac{\partial}{\partial
    \alpha_i}\frac{\partial}{\partial \alpha_j}+\sum
  _{i=1}^{\ell}\sum_{m\ge 1}b_{im} \frac{\partial}{\partial \alpha_i}
  \frac{\partial}{\partial x_{m}}+ \\+\sum_{m\ge 1} \sum
  _{i=1}^{\ell}c_{mi}\frac{\partial}{\partial x_{m}}
  \frac{\partial}{\partial \alpha_i} +
 \sum_{m,n\ge1}q_{mn}\frac{\partial^{2}}{\partial x_{m}\partial
    x_{n}}.
\end{multline}
\end{lem}
\begin{proof}
Just as in the proof of Lemma \ref{lem:eQis EQ} we use the operator description
of the coproduct by exponential operators. So we have for $a\in V$
\[
\EQ(a)
=r(e^{\sum_i\alpha_{i}\frac{\partial}{\partial{\alpha_{i}}}+\sum_{m}x_{m}
\frac{\partial}{\partial  x_{m} }}\ten
e^{\sum_j\alpha_{j}\frac{\partial}{\partial{\alpha_{j}}}+\sum_{n}x_{n}
\frac{\partial}{\partial  x_{n} }}
)(a),
\]
Now the basic point
is that $e^{\alpha_{i}\partial_{\alpha_{i}}}$ and
$e^{x_{m}\partial_{x_{m}}}$ behave like grouplike elements in
bicharacters:
we have
\begin{align*}
  r(e^{\alpha_{i}\partial_{\alpha_{i}}}\ten
  ab)&=r(e^{\alpha_{i}\partial_{\alpha_{i}}}\ten a)
  r(e^{\alpha_{i}\partial_{\alpha_{i}}}\ten b),\\
  r(e^{x_{m}\partial_{x_{m}}}\ten ab)&=r(e^{x_{m}\partial_{x_{m}}}\ten
  a)r(e^{x_{m}\partial_{x_{m}}}\ten b) ,
\end{align*}
and similar for the second argument of the bicharacter. This implies
that
\begin{align*}
  \EQ(a)&=\prod_{i,j,n,m} r(e^{\alpha_{i}\partial_{\alpha_{i}}}\ten
  e^{\alpha_{j}\partial_{\alpha_{j}}})
r(e^{\alpha_{i}\partial_{\alpha_{i}}}\ten
  e^{x_{n}\partial_{x_{n}}}) \cdot\\
&\qquad \qquad\cdot
r(e^{x_{m}\partial_{x_{m}}}\ten
  e^{\alpha_{j}\partial_{\alpha_{j}}})
r(e^{x_{m}\partial_{_{m}}}\ten
  e^{x_{n}\partial_{x_{n}}}) (a).
\end{align*}
Now one easily checks, using the values of the bicharacter on
generators, see \eqref{eq:4B}, that
\begin{align*}
  r(e^{\alpha_{i}\partial_{\alpha_{i}}}\ten
e^{\alpha_{j}\partial_{\alpha_{j}}})&=
e^{a_{ij}\partial_{\alpha_{i}}\partial_{\alpha_{j}}},\\
  r(e^{\alpha_{i}\partial_{\alpha_{i}}}\ten
  e^{x_{n}\partial_{x_{n}}}) &=
e^{b_{in}\partial_{\alpha_{i}}\partial_{x_{n}}},\\
r(e^{x_{m}\partial_{x_{m}}}\ten
  e^{\alpha_{j}\partial_{\alpha_{j}}}) &=
e^{c_{mj}\partial_{x_{m}}\partial_{\alpha_{j}}},\\
  r(e^{x_{m}\partial_{x_{m}}}\ten
  e^{x_{n}\partial_{_{n}}}) &=
e^{q_{mn}\partial_{x_{m}}\partial_{x_{n}}}.
\end{align*}
For instance, to check the first equality consider a joint
eigenspace of $\partial_{\alpha_{i}}$ and $\partial_{\alpha_{j}}$. It
consists of elements $e^{\alpha}P(x)$, with
$\alpha=m_{i}\alpha_{i}+m_{j}\alpha_{j}+\sum_{k\ne
  i,j}m_{k}\alpha_{k}$, for fixed $m_{i},m_{j}$. Then
\[
r(e^{\alpha_{i}\partial_{\alpha_{i}}}\ten
e^{\alpha_{j}\partial_{\alpha_{j}}})e^{\alpha}P(x)=r(e^{m_{i}\alpha_{i}}\ten
e^{m_{j}\alpha_{j}})e^{\alpha}P(x)=
e^{m_{i}m_{j}a_{ij}}e^{\alpha}P(x).
\]
On the other hand
\[
e^{a_{ij}\partial_{i}\partial_{j}}e^{\alpha}P(x)=e^{a_{ij}m_{i}m_{j}}e^{\alpha}P(x),
\]
proving the first equality. The other equalities are proved similarly, which proves the lemma.
\end{proof}

\bibliographystyle{alpha}

\def\cprime{$'$}

\end{document}